\documentclass[authoryear,a4paper,12pt]{article}
\usepackage[english]{babel}
\usepackage{amsmath}
\usepackage{latexsym}
\usepackage{amssymb}
\usepackage{graphicx}
\usepackage{hyperref}
\usepackage{natbib}
\usepackage{subfigure}
\usepackage{epsfig}
\newcommand{\mZ}{\mathbb{Z}}
\newcommand{\mM}{\mathbb{M}}
\newcommand{\mF}{\mathbb{F}}
\newcommand{\mG}{\mathbb{G}}
\newcommand{\mH}{\mathbb{H}}

\newcommand{\mJ}{\mathbb{J}}
\newcommand{\mP}{\mathbb{P}}
\newcommand{\mQ}{\mathbb{Q}}
\newcommand{\mR}{\mathbb{R}}
\newcommand{\diag}{\mathrm{diag}}
\def\T{{\mathrm{\scriptscriptstyle T}}}
\newcommand{\g}[1]{\gamma_{#1}}
\newcommand{\m}[1]{\mu_{#1}}
\newcommand{\p}[1]{\pi_{#1}}
\newcommand{\G}{\mathcal{G}}
\newcommand{\I}{\mathcal{I}}
\newcommand{\tD}{\tilde{D}}
\newcommand{\Mob}{M\"{o}bius}
\newcommand{\ci}{\!\perp \! \! \! \perp\!}

\newcommand{\M}[1]{\Gamma{}(#1)}
\newtheorem{definition}{Definition}
\newtheorem{lemma}{Lemma}
\newtheorem{proposition}{Proposition}
\newtheorem{theorem}{Theorem}
\newtheorem{corollary}{Corollary}
\newenvironment{proof}{\begin{trivlist}\item[] \mbox{\textit{Proof.}}}
{\hfill$\Box$ \end{trivlist}}
\begin{document}
\title{Log-mean linear models for binary data}
\author{Alberto Roverato\footnote{University of Bologna
(\texttt{alberto.roverato@unibo.it})}\\
Monia Lupparelli\footnote{University of Bologna
(\texttt{monia.lupparelli@unibo.it})}\\
Luca La Rocca\footnote{University of Modena and
Reggio Emilia (\texttt{luca.larocca@unimore.it})}}
\date{June 2, 2012}
\maketitle
\begin{abstract}\small
This paper introduces a novel class of models for binary data, which we call log-mean linear models. The characterizing feature of these models is that they are specified by linear constraints on the log-mean linear parameter, defined as a log-linear expansion of the mean parameter of the multivariate Bernoulli distribution. We show that marginal independence relationships between variables can be specified by setting certain log-mean linear interactions to zero and, more specifically, that graphical models of marginal independence are log-mean linear models. Our approach overcomes some drawbacks of the existing parameterizations of graphical models of marginal independence.
\bigskip

\end{abstract}
\section{Introduction}
A straightforward way to parameterize the probability distribution
of a set of categorical variables is by means of their probability table.
Probabilities are easy to interpret but have the drawback that sub-models of interest
typically involve non-linear constraints on these parameters.
For instance, conditional independence relationships can be specified
by requiring certain factorizations of the cell probabilities;
see \citet{lau:1996} and \citet{cw:1996}. For this reason,
it is useful to develop alternative parameterizations such that sub-models of interest
correspond to linear sub-spaces of the parameter space of the saturated model.
In particular, we are interested in
graphical models of marginal independence,
which were introduced by \citet{cw:1993,cw:1996}
with the name of covariance graph models, but later addressed in the literature
also as bidirected graph models following \citet{ric:2003}.
These models have appeared in several applied contexts
as described in \citet{drton2008binary} and references therein.

In this paper we consider binary data and introduce a novel parameterization based on
a log-linear expansion of the mean parameter of the multivariate Bernoulli distribution,
which we call the log-mean linear parameterization.
We then define the family of log-mean linear models obtained by imposing linear constraints
on the parameter space of the saturated model. We show that marginal independence
between variables can be specified by setting certain log-mean linear interactions
to zero and, more specifically, that graphical models of marginal independence are
log-mean linear models.

In the discrete case, two alternative parameterizations of bidirected graph models
are available: the \Mob\ parameterization \citep{drton2008binary}
and  the multivariate logistic parameterization \citep{glo-mcc:1995,lup-mar-ber:2009}.
Our approach
avoids some disadvantages of both these parameterizations:
log-mean linear interactions can be interpreted as measures of association,
which allows one to specify interesting sub-models not readily available
using the \Mob\ parameterization,
and the likelihood function can be written in closed form,
which is not possible with the multivariate logistic parameterization.
Furthermore, the log-mean linear approach to bidirected graph modelling is computationally more efficient than the multivariate logistic one.

\section{Preliminaries}\label{SEC:preliminaries}
%
\subsection{Parameterizations for binary data}
Given the finite set $V=\{1,\ldots, p\}$, with $\mid{}V\mid{}=p$, let $X_{V}=(X_{v})_{v\in V}$ be a random vector of binary variables taking values in the set $\I_V=\{0,1\}^{p}$. We call $\I_V$ a $2^p$-table and its elements $i_{V} \in \I_V$ the cells of the table. In this way, $X_{V}$~follows a multivariate Bernoulli distribution with probability table $\pi(i_{V})$, $i_{V}\in\I_{V}$, which we assume to be strictly positive. Since $\I_{V}=\{0, 1\}^{p}=\{(1_{D}, 0_{V \backslash D})\mid D\subseteq V\}$, we can write the probability table as a vector $\pi=(\p{D})_{D \subseteq V}$ with entries $\pi=\mbox{pr}(X_{D}=1_{D},X_{V\backslash D}=0_{V\backslash D})$. We refer to $\pi$ as to the probability parameter of $X_{V}$ and recall that it belongs to the $(2^{p}-1)$-dimensional simplex, which we write as $\pi\in \Pi$.

In general,
we call $\theta$ a parameter of $X_{V}$ if it is a vector in $R^{2^p}$ that characterizes the joint probability distribution of $X_{V}$, and use the convention that the entries of $\theta$ (called interactions) are indexed by the subsets of $V$, i.e., $\theta=(\theta_D)_{D \subseteq V}$. If $\omega$ is an alternative parameter of $X_{V}$, then a result known as \Mob\ inversion states that
\begin{equation}\label{EQN.mon.Mobius}
\omega_D=\sum_{E \subseteq D} \theta_E\quad  (D \subseteq V)\quad
\Longleftrightarrow \quad
\theta_D=\sum_{E \subseteq D} (-1)^{|D\backslash E|} \omega_E\quad (D \subseteq V);
\end{equation}
see, among others, \citet[][Appendix~A]{lau:1996}.
Let $Z$ and $M$ be two $(2^p \times 2^p)$ matrices
with entries indexed by the subsets of $V \times V$
and given by $Z_{D,H}= 1 (D \subseteq H)$
and $M_{D,H}= (-1)^{|H \backslash D|} 1(D \subseteq H)$,
respectively, where $1(\cdot)$ denotes the indicator function.
Then, the equivalence~(\ref{EQN.mon.Mobius}) can be written in matrix form
as $\omega=  Z^{\T} \theta$ if and only if $\theta =M^{\T}\omega$,
and \Mob\ inversion follows by noticing that $M= Z^{-1}$.

We now review some well-known alternative parameterizations for the distribution of $X_V$, each defined by a smooth invertible mapping from $\Pi$ onto a smooth $(2^{p}-1)$-dimensional manifold of $R^{2^p}$. For simplicity, we denote both the mapping and the alternative parameter it defines by the same (greek) letter.

Multivariate Bernoulli distributions form a regular exponential family
with canonical log-linear parameter $\lambda$ computed as
$\lambda=M^{T} \log \p{}$.
The parameterization $\lambda$ captures conditional features
of the distribution of $X_V$ and is used to define the class of log-linear models,
which includes as a special case the class of undirected graphical models;
see~\citet[Chap.~4]{lau:1996}.

The mean parameter of the multivariate Bernoulli
distribution is $\m{}=(\m{D})_{D \subseteq V}$, where $\m{\emptyset}=1$
(on grounds of convention) and $\m{D}=P(X_{D}=1_{D})$ otherwise.
This was called the \Mob\ parameter by \citet{drton2008binary},
because one finds $\m{}= Z\p{}$.  The linear mapping $\pi\mapsto\mu$ is
trivially \Mob-inverted to obtain $\p{}=M \m{}$, for all $\m{}
\in \m{}(\Pi)$. However, the structure of $\m{}(\Pi)$ is rather involved,
and actually well-understood only for small~$p$.
The parameterization $\m{}$ captures marginal distributional features
of $X_V$ and thus satisfies the upward compatibility property, i.e.,
it is invariant with respect to marginalization.

\citet{Ekh-al:1995}, in a context of regression analysis,
proposed to modify the mean parameter by replacing each entry
$\mu_{D}$ of $\mu$ such that $\mid{}D\mid>1$ with the corresponding dependence ratio
defined as $\tau_{D}=\mu_{D}/(\prod_{v\in D}\mu_{\{v\}})$;
see also \citet{Ekh-al:2000} and \citet{dar-spe:1983},
where these ratios were used in models named Lancaster additive.
We define $\tau_D=\mu_D$ for $\mid D\mid\le1$
and call $\tau=(\tau_D)_{D\subseteq V}$ the dependence ratio parameter.

\citet{ber-rud:2002} developed
a wide class of parameterizations capturing
both marginal and conditional distributional features,
named marginal log-linear parameterizations,
which have been applied in several contexts;
see \citet{ber-al:2009}. Broadly speaking,
any marginal log-linear parameter is obtained by stacking
subvectors of log-linear parameters computed in suitable marginal
distributions. This class of parameterizations includes as
special, extreme, cases the log-linear parameterization $\lambda$,
where a single margin is used, and the multivariate
logistic parameterization of \citet{glo-mcc:1995},
denoted by $\eta=(\eta_D)_{D\subseteq V}$, where each $\eta_D$ is computed in
the margin $X_D$. The parameterization $\eta$ clearly
satisfies the upward compatibility property,
while the structure of $\eta(\Pi)$ is rather involved.
A disadvantage of these 
parameterizations is that their inverse mappings cannot be
analytically computed (but for the special case of $\lambda$).
\subsection{Bidirected graph models}\label{SUB.luc.bigraphs}
\begin{figure}
\begin{center}
\includegraphics[scale=1]{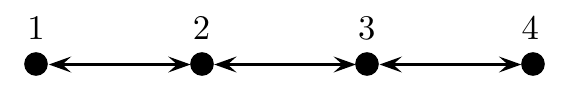}
\end{center}
\caption{Bidirected graph with disconnected sets
$\{1,3\}$, $\{1,4\}$, $\{2,4\}$, $\{1,2,4\}$ and $\{1,3,4\}$,
encoding the independencies $X_{\{1,2\}} \ci X_{4}$ and
$X_{1} \ci X_{\{3,4\}}$.}
\label{FIG.bid.graph}
\end{figure}

Graphical models of marginal independence
aim to capture marginal independence relationships
between variables.
Following \citet{ric:2003},
we use the convention that the independence structure of variables
is represented by a bidirected graph.
Nevertheless, we recall that these same models
have been previously discussed by \citet{cw:1993}
adopting a different graphical representation with undirected dashed edges.

A bidirected graph $\G=(V, E)$ is defined by a set $V = \{1, \dots, p\}$ of nodes
and a set $E$ of edges drawn as bidirected.
A set $D \subseteq V$ is said to be connected in $\G$ if it induces a connected subgraph
and it is said to be disconnected otherwise.
Any disconnected set $D \subseteq V$ can be uniquely partitioned into
its connected components $C_1, \dots, C_r$ such that $D = C_1\cup \cdots\cup C_r$;
see \citet{ric:2003} for technical details.

A bidirected graph model is the family of probability distributions for $X_V$
satisfying a given Markov property with respect to a bidirected graph $\G$.
The distribution of $X_V$ satisfies
the connected set Markov property \citep{ric:2003} if,
for every disconnected set~$D$,
the subvectors corresponding to its connected components $X_{C_1},\dots,X_{C_r}$
are mutually independent; in symbols $X_{C_{1}}\ci X_{C_{2}}\ci\cdots\ci X_{C_{r}}$.
We denote by $B(\G)$ the bidirected graph model for $X_V$
defined by $\G$ under the connected set Markov property.
See Figure~\ref{FIG.bid.graph} for an example.

Parameterizations for the class $B(\G)$  have been studied by
\citet{drton2008binary} and \citet{lup-mar-ber:2009},
where $B(\G)$ is defined by imposing
multiplicative constraints on $\m{}$ and linear constraints on $\eta$,
respectively; see also \citet{Forc-al:2010},
\citet{rud-al:2010}, \citet{eva-rich:2012} and \citet{mar-lup:2011}.

\section{Log-mean linear models}\label{SEC.mon.model}
\label{sec:LMLmodels}
We introduce a new class of models for the multivariate Bernoulli distribution
based on the notion of log-mean linear parameter,
denoted by $\g{}=(\g{D})_{D \subseteq V}$. Each element $\g{D}$ of $\g{}$
is a log-linear expansion of a subvector, namely $(\mu_E)_{E \subseteq D}$,
of the mean parameter:
\begin{eqnarray}\label{EQN:gammaD}
\g{D}=\sum_{E\subseteq D}\; (-1)^{\mid D\backslash E\mid}\;\log(\m{E}),
\end{eqnarray}
so that in vector form we have $\g{}=M^{\T}\log \m{}$.
Notice that by replacing $\m{}$ with $\p{}$ in~(\ref{EQN:gammaD})
one obtains the canonical log-linear parameter $\lambda$.
Indeed, we will show in the next section that the log-mean linear parameterization
defines a parameter space where the multiplicative constraints on the \Mob\ parameter
of \citet{drton2008binary} correspond to linear sub-spaces and,
from this perspective, it resembles the connection
between log-linear interactions and  cell probabilities.
It should be stressed, however, that the parameter space where $\g{}$ lives has,
like $\mu(\Pi)$ and $\eta(\Pi)$, a rather involved structure.

It is worth describing in detail the elements of $\g{}$ corresponding to sets
with low cardinality (its low-order interactions). Firstly, and trivially,
$\g{\emptyset}=\log \m{\emptyset}$ is always zero. Secondly,
for every $j \in V$, the main log-mean linear effect  $\g{\{j\}}=\log \m{\{j\}}$ is always negative, because $\m{{\{j\}}}$ is a probability. Then, for every $j,k \in V$,
the two-way log-mean linear interaction $\g{\{j,k\}}=\log\{\m{\{j,k\}}/(\m{\{j\}}\m{\{k\}})\}$
coincides with the logarithm of the second-order dependence ratio.
Finally, for every triple $j,k,z \in V$, the three-way
interaction is
$$
\g{\{j,k,z\}}= \log \frac{   \m{\{j,k,z\}} \m{\{j\}}  \m{\{k\}}
 \m{\{z\}}  }{ \m{\{j,k\}}  \m{\{j,z\}} \m{\{k,z\}}}
$$
and thus differs from the third-order dependence ratio;
the same is true for each $\g{D}$
with $\mid D\mid \geq 3$. Note that, already from two-way log-mean linear interactions,
it is apparent that $\g{}$ is not a marginal log-linear parameter of \citet{ber-rud:2002}.

We now formally define the log-mean linear parameterization as a mapping
from $\Pi$.
\begin{definition}\label{DEF:LML.parameter}
For a vector $X_V$ of binary variables, the log-mean linear parameterization $\gamma$ is defined by the mapping
\begin{eqnarray}\label{EQN.mon.map_pi-ga}
\g{}=M^{\T}\log  Z \p{},\quad\p{}\in\Pi.
\end{eqnarray}
\end{definition}
The multivariate logistic parameter $\eta$ can also be computed as
$\eta=C\log(L\pi)$ for a suitable choice of matrices $C$ and $L$,
so that the mapping $\pi\mapsto\eta$  resembles (\ref{EQN.mon.map_pi-ga}),
but with the major difference that $C$ and $L$ are rectangular matrices
of size $t\times 2^p$, with $t\gg 2^p$, so~that the inverse transformation
is not available in closed form.
On the other hand, in our case the inverse transformation can be analytically computed
by applying \Mob\ inversion twice to obtain $\p{}= M \exp  Z^{\T} \g{}$.
Clearly, the bijection specified by $\pi\mapsto\gamma$ is smooth,
so~that it constitutes a valid reparameterization. Finally, like $\mu$ and $\eta$, the parameterization $\gamma$ satisfies the upward compatibility property.

We next define log-mean linear models as follows.
\begin{definition}
For a vector $X_V$ of binary variables and a full rank $(2^p \times k)$ matrix~$H$,
where $k < 2^p$ and the rows of $H$ are indexed by the subsets of $V$,
the log-mean linear model $\M{H}$ is the family of probability distributions
for $X_V$ such that $H^{T}\g{}=0$.
\label{DEF:LML.model}
\end{definition}
It is not difficult to construct a matrix $H$ such that $\M{H}$ is empty. However, the family $\M{H}$ is non-empty if the linear constraints neither involve $\g{\emptyset}$ nor the main effect $\g{\{j\}}$, for every $j \in V$. More formally, a sufficient condition for $\M{H}$ to be non-empty is that the rows of $H$ indexed by $D \subseteq V$ with $\mid D\mid \leq 1$ be all equal to zero; see \S~\ref{SEC.mon.bid}.
\begin{proposition}
Any non-empty log-mean linear model $\M{H}$ is
a curved exponential family of dimension $(2^p-k-1)$.
\end{proposition}
\begin{proof}
This follows from the mapping defining the parameterization
$\gamma$ being smooth, and the matrix $H$ imposing a
$k$-dimensional linear constraint on the parameter $\gamma$.
\end{proof}
Maximum likelihood estimation for log-mean linear models under a Multinomial or Poisson sampling scheme is a constrained optimization problem, which can be solved by means of standard algorithms. Specifically, we adopt an iterative method typically used for fitting marginal log-linear models which also gives the asymptotic standard errors;
see Appendix~\ref{SEC:APP-C} for details.
In our case, the algorithm is computationally more efficient than for marginal log-linear models, especially when these are obtained by constraining the multivariate logistic parameter, because, as remarked above, rectangular matrices of size $t\times 2^p$ with $t\gg 2^p$ are replaced by square matrices of size $2^p\times 2^p$.

The elements of $\gamma$, as well as those of $\mu$
and of $\tau$, are not symmetric under relabelling of the two states taken by the random variables,
because they measure event specific association.
\citet[\S~4]{Ekh-al:1995} show that in some contexts this feature may
amount to an advantage; see also the application in \S~\ref{sec:examples}.
Furthermore, this is not an issue in
the definition of bidirected graph models,
which is illustrated in the next section.

\section{Log-mean linear models and marginal independence}
\label{SEC.mon.bid}
We show that the log-mean linear parameterization $\g{}$ can be used to encode marginal independencies and, also, that bidirected graph models are log-mean linear models. Hence, the log-mean linear parameterization can be used in alternative to the approaches developed by \citet{drton2008binary} and \citet{lup-mar-ber:2009}. Our approach is appealing because it combines the advantages of the \Mob\ parameterization~$\m{}$ and of the multivariate logistic parameterization~$\eta$: the inverse map $\g{} \mapsto \p{}$ can be analytically computed, as for $\m{}$, and the model is defined by means of linear constraints, as for $\eta$.

The following theorem shows how suitable linear constraints on the log-mean linear parameter correspond to marginal independencies; see Appendix~\ref{SEC:APP-B} for a proof.
\begin{theorem}\label{THM:gamma.eq.zero.and.fact.mu}
For a vector $X_{V}$ of binary variables with probability parameter $\p{}\in\Pi$,
let~$\m{}=\m{}(\p{})$ and $\g{}=\g{}(\p{})$. Then, for a pair of disjoint, nonempty,
proper subsets $A$ and~$B$ of $V$, the following conditions are equivalent:
\begin{itemize}
\item[(i)] $X_{A}\ci X_{B}$; \item[(ii)] $\m{A^{\prime}\cup
B^{\prime}}=\m{A^{\prime}}\times \m{B^{\prime}}$ for every
$A^{\prime}\subseteq A$ and $B^{\prime}\subseteq B$; \item[(iii)]
$\g{A^{\prime}\cup B^{\prime}}=0$ for every $A^{\prime}\subseteq
A$ and $B^{\prime}\subseteq B$ such that $A^{\prime}\neq\emptyset$
and $B^{\prime}\neq \emptyset$.
\end{itemize}
\end{theorem}
We remark that the equivalence (i)$\Leftrightarrow$(ii)
of Theorem~\ref{THM:gamma.eq.zero.and.fact.mu} follows immediately
from Theorem~1 of \citet{drton2008binary}. Furthermore,
it is straightforward to see that (ii) could be restated by replacing
the $\mu$-interactions with the corresponding $\tau$-interactions.

The next result generalizes Theorem~\ref{THM:gamma.eq.zero.and.fact.mu}
to the case of three or more subvectors; see Appendix~\ref{SEC:APP-B} for a proof.
\begin{corollary}\label{THM:gamma.eq.zero.and.mutual.indep}
For a sequence $A_{1},\ldots, A_{r}$ of $r \geq 2$ pairwise disjoint, nonempty,
subsets of $V$, let $\mathcal{D}=\{D\mid D\subseteq A_{1}\cup\cdots\cup A_{r}
\;\mbox{with}\;D\not\subseteq A_{i}\;\mbox{for}\;i=1,\ldots,r\}$.
Then $X_{A_{1}},\dots, X_{A_{r}}$ are mutually independent if and only if
$(\g{D})_{D\in\mathcal{D}}=0$.
\end{corollary}
An interesting special case of~Corollary~\ref{THM:gamma.eq.zero.and.mutual.indep}
is given below; see Appendix~\ref{SEC:APP-B} for a proof.
\begin{corollary}\label{THM:independence.of.singletons}
For a subset $A\subseteq V$ with $\mid A\mid>1$, the variables in $X_{A}$
are mutually independent  if and only if $\g{D}=0$ for every
$D\subseteq A$ such that $\mid D\mid>1$.
\end{corollary}
We stated in \S~\ref{SEC.mon.model} that $\M{H}$ is
non-empty whenever the rows indexed by $D\subseteq V$ with
$\mid D\mid\le1$ are equal to zero. This fact derives from
Corollary~\ref{THM:independence.of.singletons},
because the distribution of mutually independent variables satisfies
the constraint $H^T\g{}=0$.

It follows from Theorem~\ref{THM:gamma.eq.zero.and.fact.mu} that
the probability distribution of $X_{V}$ satisfies the pairwise Markov property
with respect to a bidirected graph $\G=(V, E)$ if and only if $\g{\{j,k\}}=0$
whenever $j$ and $k$ are disjoint nodes in $\G$.
The following theorem shows that bidirected graph models for binary data
are log-mean linear models also under the connected set Markov property;
see Appendix~\ref{SEC:APP-B} for a proof.
\begin{theorem}\label{THM:gamma.eq.zero.and.disc.MP}
The distribution of a vector of binary variables $X_V$ belongs to
the bidirected graph model $B(\G)$ if and only if its log-mean
linear parameter $\g{}$ is such that $\g{D}=0$ for every set $D$
disconnected in $\G$.
\end{theorem}
For instance, if $\G$ is the graph in Figure~\ref{FIG.bid.graph}
the bidirected graph model $B(\G)$ is defined by the linear constraints $\g{\{1,3\}}=\g{\{1,4\}}=\g{\{2,4\}}=\g{\{1,2,4\}}=\g{\{1,3,4\}}=0$.

\section{Application}\label{sec:examples}
%

Table~\ref{tab.coppen} shows data from \citet{coppen:1966}
for a set of four binary variables concerning symptoms of 362 psychiatric patients.
\cite{wermuth:1976} analysed these data within the family of
decomposable undirected graphical models,
but a visual inspection of Table~6 of \cite{wermuth:1976} suggests that
also investigating the marginal independence structure may be useful.
\begin{table}[b]
\begin{center}\small
\caption{Data from \citet{coppen:1966} on four symptoms of 362 psychiatric patients.}
\bigskip
\label{tab.coppen}
\begin{tabular}{lllrrrr} \hline
            &             & Solidity &  hysteric      &            &         rigid & \\
Stability & Depression & Validity &  psychasthenic &  energetic & psychasthenic & energetic \\
\hline
extroverted &          no &   &       12 &          47 &        8 &          14  \\
            &         yes &   &       16 &          14 &       22 &          23  \\
introverted &          no &   &       27 &          46 &       22 &          25  \\
            &         yes &   &       32 &           9 &       30 &          15  \\
\hline
\end{tabular}
\end{center}
\end{table}
For this reason, we performed an exhaustive model search
within the family of bidirected graph models and selected
the model with optimal value of the Bayesian information criterion
among those whose $p$-value,
computed on the basis of the asymptotic chi-squared distribution of the deviance,
is not smaller than 0.05.
The selected model has deviance $\chi^2_{(5)}=8.6$ ($p=0.13$, $BIC=-20.85$) and
corresponds to the graph of Figure~\ref{FIG.bid.graph},
where $X_{1}=\mbox{Stability}$, $X_{2}=\mbox{Validity}$,
$X_{3}=\mbox{Depression}$ and $X_{4}=\mbox{Solidity}$.

The application of bidirected graph models is typically motivated
by the fact that the observed Markov structure can be
represented by a data generating processes with latent variables.
In particular, the independence structure of the selected model is compatible,
among others, with the generating process represented in Figure~\ref{FIG:DAG.with.latent},
where $U$ is a latent factor.
\begin{figure}
\begin{center}
\includegraphics[scale=1]{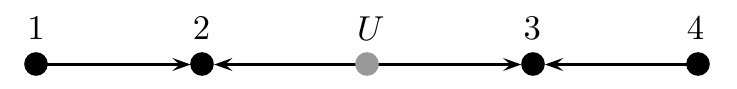}
\end{center}
\caption{A generating model for Coppen's data.}
\label{FIG:DAG.with.latent}
\end{figure}
Under the generating model in Figure~\ref{FIG:DAG.with.latent},
one may be interested in investigating substantive research hypotheses
on the role of the latent. For instance,
$U$ might be a binary variable representing a necessary factor for Depression:
$\{U=\mbox{on}\}$ might be a necessary condition for patients
to have acute depression, that is, for $\{X_{3}=\mbox{yes}$\}.
Formally, we might have $\mbox{pr}(X_{3}=\mbox{yes}\mid U=\mbox{off})=0$,
whereas $0<\mbox{pr}(X_{3}=\mbox{yes}\mid U=\mbox{on})<1$; see \citet[][\S~3.1]{Ekh-al:2000}.

If, in the above generating process, $U$ represents a necessary factor, then the context-specific independence $X_{\{1,2\}} \ci X_{4} \mid \{X_{3}=\mbox{yes}\}$ holds but, typically, $X_{\{1,2\}} {\not\ci} X_{4} \mid \{X_{3}=\mbox{no}\}$. Furthermore, if the levels of
$X_3=\mbox{Depression}$ are coded so that  $\mbox{yes}=1$, the above context-specific independence is satisfied in the selected marginal independence model if and only if some additional log-mean linear interactions are equal to zero, namely, $\gamma_{\{2,3,4\}}=\gamma_{\{1,2,3,4\}}=0$; see Appendix~\ref{SEC:APP-D} for details. Thus, by properly coding the levels of $X_{3}$, we can specify a log-mean linear model that encodes the independence structure of the graph in Figure~\ref{FIG:DAG.with.latent} together with the context-specific independence implied by the assumption that $U$ is a necessary factor for $\{X_{3}=\mbox{yes}\}$. This model has deviance $\chi^2_{(7)}=17.08$ ($p=0.02$, $BIC=-24.16$) and, therefore, the necessary factor hypothesis is only weakly supported by the data. We remark that this log-mean linear model is invariant with respect to the coding of $X_{\{1,2,4\}}$, because it is fully specified by the constraints $X_{\{1,2\}} \ci X_{4}$, $X_{1} \ci X_{\{3,4\}}$ and  $X_{\{1,2\}}\ci X_{4}\mid \{X_{3}=1\}$. On the other hand, the log-mean linear model specified by the same zero constraints, but coding the levels of $X_{3}=\mathrm{Depression}$ so that $\mbox{no}=1$, allows one to verify the hypothesis that $U$ is a necessary factor for the absence of depression, that is, for $\{X_{3}=\mbox{no}\}$. The latter log-mean linear model provides an adequate fit with deviance $\chi^2_{(7)}=9.3$ ($p=0.23$, $BIC=-31.94$) so that the hypothesis is not contradicted by the data.

\section{Discussion}\label{SEC:discussion}
Our log-linear expansion of $\mu$ provides the first instance of a parameterization for binary data, not belonging to the class of marginal log-linear parameterizations, which allows one to specify bidirected graph models through linear constraints.

We deem that  the log-mean linear parameterization represents an appealing candidate for the implementation of Bayesian procedures for this class of models because the likelihood function under Multinomial or Poission sampling is explicitly available and marginal independencies correspond to zero-interactions. However, there are still difficulties related to the involved structure of the parameter space, which is a common trait of marginal parameterizations.

The specification of log-mean linear models encoding substantive research hypotheses, possibly by exploiting the asymmetry of our parameterization with respect to variable coding which we briefly touched upon in \S~\ref{sec:examples}, represents an open research area. Clearly, log-mean linear models can incorporate any linear constraint on $\log(\tau)$, because the latter is a linear transformation of $\gamma$. Some instances of substantive research assumptions that can be expressed in this way, such as, for instance, horizontal and vertical homogeneity of dependence ratios, can be found in \citet{Ekh-al:1995} and \citet{Ekh-al:2000}.

\section*{Acknowledgments}
We gratefully acknowledge useful discussions with David R. Cox, Mathias Drton, Antonio Forcina,  Giovanni M. Marchetti, and Nanny Wermuth.
\appendix
\section*{Appendices}
\section{Proofs of technical results}\label{SEC:APP-B}
The following Lemma is instrumental in proving Theorem~\ref{THM:gamma.eq.zero.and.fact.mu}.
%
%
\begin{lemma}\label{THM:sum.to.zero.more.general}
Let $g(\cdot)$ be a real-valued function defined on the sub-sets of
a set $D$. If two non-empty, disjoint, proper sub-sets $A$ and $B$
of $D$ exist, such that $A\cup B=D$ and
$g(E)=g(E\cap A)+g(E\cap B)$ for every $E\subseteq D$,
then $\sum_{E\subseteq D}\;(-1)^{|D\backslash E|}\;g(E)=0$.
\label{lemma}
\end{lemma}
\begin{proof}
We start this proof by recalling a well-known fact. It can be proven by induction that any non-empty set $D$
has the same number of even and odd sub-sets. Consequently, it holds that 
\begin{equation}
\sum_{E\subseteq
D}\;(-1)^{|E|}=\sum_{E\subseteq D}\;(-1)^{|D\backslash E|}=0\quad \mbox{for all set } D\neq\emptyset.
\label{eq:basic}
\end{equation}
We will use this fact twice in the remainder of this proof. 

If we set $h=\sum_{E\subseteq D}\;(-1)^{|D\backslash E|}\;g(E)$,
then we have to show that $h=0$. Since $A$ and $B$ form a
partition of $D$, we can write
\begin{eqnarray*}
h &=&\sum_{A^{\prime}\subseteq A}\sum_{B^{\prime}\subseteq B}\;
(-1)^{|(A\cup B)\backslash (A^{\prime}\cup B^{\prime})|}\;
g(A^{\prime}\cup B^{\prime}),
\end{eqnarray*}
where $A^\prime=E\cap A$ and  $B^\prime=E\cap B$. Then, from the
fact that $A\cap B=A^{\prime}\cap B^{\prime}=A^{\prime}\cap
B=B^{\prime}\cap A=\emptyset$ it follows both that $(-1)^{|(A\cup
B)\backslash (A^{\prime}\cup B^{\prime})|}= (-1)^{|A\backslash
A^{\prime}|}\times (-1)^{|B\backslash B^{\prime}|}$ and that
$g(A^{\prime}\cup B^{\prime})=g(A^{\prime})+g(B^{\prime})$. Hence,
we obtain
\begin{eqnarray*}
h&=&\sum_{A^{\prime}\subseteq A}\sum_{B^{\prime}\subseteq B}\;(-1)^{|A\backslash A^{\prime}|}(-1)^{|B\backslash B^{\prime}|} \; \left\{g(A^{\prime})+g(B^{\prime})\right\}\\
&=& \sum_{A^{\prime}\subseteq A}\; (-1)^{|A\backslash A^{\prime}|} \sum_{B^{\prime}\subseteq B} (-1)^{|B\backslash B^{\prime}|}\; \left\{g(A^{\prime})+g(B^{\prime})\right\}\\
&=& \sum_{A^{\prime}\subseteq A}\; (-1)^{|A\backslash
A^{\prime}|}\; \left\{g(A^{\prime})\sum_{B^{\prime}\subseteq B}
(-1)^{|B\backslash B^{\prime}|} +\sum_{B^{\prime}\subseteq B}
(-1)^{|B\backslash B^{\prime}|}g(B^{\prime})\right\}.
\end{eqnarray*}
By assumption $B\neq\emptyset$, so that equation (\ref{eq:basic}) implies
$\sum_{B^{\prime}\subseteq B} (-1)^{|B\backslash B^{\prime}|}=0$ and thus
\begin{eqnarray*}
h&=& \sum_{A^{\prime}\subseteq A}\; (-1)^{|A\backslash
A^{\prime}|}\; \left\{\sum_{B^{\prime}\subseteq B}
(-1)^{|B\backslash B^{\prime}|}g(B^{\prime})\right\}.
\end{eqnarray*}
Since we also have $A\neq\emptyset$,
equation (\ref{eq:basic}) also implies that
$\sum_{A^{\prime}\subseteq A}\; (-1)^{|A\backslash A^{\prime}|}=0$
and therefore that $h=0$, as required.
\end{proof}

\subsection*{Proof of Theorem~\ref{THM:gamma.eq.zero.and.fact.mu}}

We first show (i)$\Leftrightarrow$(ii). The implication
(i)$\Rightarrow$(ii) is straightforward. To prove that
(i)$\Leftarrow$(ii) we use the same argument as in the proof of
Theorem~1 of \citet{drton2008binary}, which for completeness we
now give in detail.

We want to show that for every $i_{A\cup B}\in\I_{A\cup B}$ it
holds that
\begin{eqnarray}\label{EQN:alb.10}
P(X_{A\cup B}=i_{A\cup B})=P(X_{A}=i_{A})P(X_{B}=i_{B})
\end{eqnarray}
and we do this by induction on the number of 0s in $i_{A\cup B}$,
which we denote by~$k$, with $0\leq k\leq |A\cup B|$. More
precisely, point (ii) implies that the factorization~(\ref{EQN:alb.10}) is satisfied for $k=0$, also when $A$ and $B$
are replaced with proper subsets, and we show that if such
factorization is satisfied for every $k<j\leq |A\cup B|$ then it
is also true for $k=j$. Since $j>0$, there exists $v\in A\cup B$
such that $i_{v}=0$ and, in the following, we assume without loss
of generality that $v\in A$, and set $A^{\prime}=A\backslash
\{v\}$. Hence,
\begin{eqnarray*}
P(X_{A\cup B}=i_{A\cup B})
&\!\!=\!\!&
P(X_{A^{\prime}\cup B}=i_{A^{\prime}\cup B})-P(X_{A^{\prime}\cup B}=i_{A^{\prime}\cup B}, X_{v}=1)\\
&\!\!=\!\!&
P(X_{A^{\prime}}=i_{A^{\prime}})P(X_{B}=i_{B}) - P(X_{A^{\prime}}=i_{A^{\prime}}, X_{v}=1)P(X_{B}=i_{B})\\
&\!\!=\!\!&
\left\{P(X_{A^{\prime}}=i_{A^{\prime}})-P(X_{A^{\prime}}=i_{A^{\prime}},
X_{v}=1)\right\}
P(X_{B}=i_{B})\\
&\!\!=\!\!&
P(X_{A}=i_{A})P(X_{B}=i_{B})
\end{eqnarray*}
as required; note that the factorizations in the second equality
follow from (ii) and the inductive assumption, because the number
of 0s in $i_{A^{\prime}\cup B}$ is $j-1$, and furthermore that for
the case $A^{\prime}=\emptyset$ we use the convention
$P(X_{A^{\prime}}=i_{A^{\prime}})=1$ and
$P(X_{A^{\prime}}=i_{A^{\prime}}, X_{v}=1)=P(X_{v}=1)$.

We now show (ii)$\Leftrightarrow$(iii). The implication
(ii)$\Rightarrow$(iii) follows by noticing that
$$\g{D}=\sum_{E\subseteq D}\;(-1)^{|D\backslash E|}\;g(E),$$
where $g(E)=\log\m{E}$. Hence, if we set $D=A^{\prime}\cup B^{\prime}$,
with $A^{\prime}$ and $B^{\prime}$ as in (iii),  the statement in
(ii) implies that for every $E\subseteq D$
\begin{eqnarray*}
g(E)=\log\m{E}=\log\m{A^{\prime}\cap E}+\log\m{B^{\prime}\cap
E}=g(A^{\prime}\cap E)+g(B^{\prime}\cap E)
\end{eqnarray*}
so that the equality $\g{D}=0$ follows immediately from
Lemma~\ref{THM:sum.to.zero.more.general}. We next show that
(ii)$\Leftarrow$(iii) by induction on the cardinality of $A\cup
B$, which we again denote by $k$.

We first notice that the identity $\m{A\cup B}=\m{A}\times\m{B}$
is trivially true whenever either $A=\emptyset$ or $B=\emptyset$
because $\m{\emptyset}=1$. Then, if $|A\cup B|=2$, so that
$|A|=|B|=1$,
$\g{A\cup B}=0$ implies $\m{A\cup B}=\m{A}\times\m{B}$ as an
immediate consequence of the identity
$\g{A\cup B}=\log\{\m{A\cup B}/\m{A}\m{B}\}$. Finally, we show
that if the result is true for $|A\cup B|<k$  then it also holds
for $|A\cup B|=k$. To this aim, it is useful to introduce the
vector $\m{}^{\ast}$ indexed by $E\subseteq A\cup B$ defined as
follows:
\begin{eqnarray*}
\m{}^{\ast}=\left\{
\begin{array}{ll}
\m{E} &\quad \mbox{for $E\subset A\cup B$};\\
\m{A}\times\m{B} &\quad \mbox{for $E=A\cup B$}.\\
\end{array}
\right.
\end{eqnarray*}
Condition (iii) is recursive and, therefore, if it is satisfied
for $A$ and $B$ then it is also satisfied for every
$A^{\prime}\subseteq A$ and $B^{\prime}\subseteq B$ such that
$|A^{\prime}\cup B^{\prime}|<k$, that is, such that
$A^{\prime}\cup B^{\prime}\subset A\cup B$. As a consequence, the
inductive assumption implies that  $\m{A^{\prime}\cup
B^{\prime}}=\m{A^{\prime}}\times \m{B^{\prime}}$ for every
$A^{\prime}\subseteq A$ and $B^{\prime}\subseteq B$ such that
$A^{\prime}\cup B^{\prime}\neq A\cup B$, and this in turn has two
implications: firstly, we only have to prove that (iii) implies
$\m{A\cup B}=\m{A}\times \m{B}$; secondly, we have
$\sum_{E\subseteq A\cup B}\; (-1)^{|(A\cup B)\backslash E|} \;
\log\m{E}^{\ast}=0$ by Lemma~\ref{THM:sum.to.zero.more.general}.
Hence, we can write
\begin{eqnarray}
\nonumber \g{A\cup B}
&=& \sum_{E\subseteq A\cup B}\; (-1)^{|(A\cup B)\backslash E|} \;\log\m{E}\\
\nonumber
&=&\log\m{A\cup B} + \sum_{E\subset A\cup B}\; (-1)^{|(A\cup B)\backslash E|} \;\log\m{E}^{\ast}\\
\nonumber
&=&\log\m{A\cup B}-\log\m{A}-\log\m{B} + \sum_{E\subseteq A\cup B}\; (-1)^{|(A\cup B)\backslash E|} \; \log\m{E}^{\ast}\\
\label{EQN:alb.04} &=&\log\m{A\cup B}-\log\m{A}-\log\m{B}
\end{eqnarray}
and since (iii) implies that $\g{A\cup B}=0$ then
(\ref{EQN:alb.04}) leads to $\m{A\cup B}=\m{A}\times \m{B}$, and
the proof is complete.

\subsection*{Proof of Corollary~\ref{THM:gamma.eq.zero.and.mutual.indep}}

For $i=1,\ldots, r$, we introduce the sets
$A_{-i}=\bigcup_{j\neq i} A_{j}$ and $\mathcal{D}_i=\{D|
D\subseteq A_{i}\cup A_{-i}, \textrm{ with both } D\cap
A_{i}\neq\emptyset\mbox{ and }D\cap A_{-i}\neq\emptyset\}$ and
note that, by Theorem~\ref{THM:gamma.eq.zero.and.fact.mu},
$X_{A_i} \ci X_{A_{-i}}$ if and only if $\g{D}=0$ for every $D \in
\mathcal{D}_i$. The mutual independence $X_{A_1} \ci \cdots \ci
X_{A_r}$ is equivalent to $X_{A_i} \ci X_{A_{-i}}$ for every $i
=1,\dots,r$ and, by Theorem~\ref{THM:gamma.eq.zero.and.fact.mu},
the latter holds true if and only if $\g{D}=0$ for every $D \in
\bigcup_{i=1}^r \mathcal{D}_i$. Hence, to prove the desired result
we have to show that $\mathcal{D} = \bigcup_{i=1}^r
\mathcal{D}_i$.

It is straightforward to see that $\mathcal{D}_{i}\subseteq
\mathcal{D}$ for every $i=1,\ldots, r$, so that
$\mathcal{D}\supseteq\bigcup_{i=1}^r \mathcal{D}_i$. The reverse
inclusion $\mathcal{D}\subseteq\bigcup_{i=1}^r\mathcal{D}_i$ can
be shown by noticing that for any $D\in\mathcal{D}$ one can always
find at least one set $A_{i}$ such that $D\cap A_{i}\neq
\emptyset$; since $D\not\subseteq A_{i}$ by construction, it holds
that $D\cap A_{-i}\neq \emptyset$ and therefore that $D\in
\mathcal{D}_{i}$. Hence, we have $D\in
\bigcup_{i=1}^r\mathcal{D}_{i}$ for every $D\in \mathcal{D}$, and
this completes the proof.

\subsection*{Proof of Corollary~\ref{THM:independence.of.singletons}}
It is enough to apply
Corollary~\ref{THM:gamma.eq.zero.and.mutual.indep} by taking
$A=A_{1}\cup\cdots\cup A_{r}$ with $|A_{i}|=1$ for every
$i=1,\ldots, r$.

\subsection*{Proof of Theorem~\ref{THM:gamma.eq.zero.and.disc.MP}}

Every set $D\subseteq V$ that is disconnected in $\G$ can be partitioned uniquely into inclusion maximal connected sets $\tD_{1},\ldots,\tD_{r}$ with $r\geq 2$. It is shown in Lemma~1 of \citet{drton2008binary} that $\p{}\in B(\G)$ if and only if $X_{\tD_{1}}\ci\cdots\ci X_{\tD_{r}}$ for every disconnected set $D\subseteq V$. Hence, it is sufficient to prove that the mutual independence $X_{\tD_{1}}\ci\cdots\ci X_{\tD_{r}}$ holds for every disconnected set $D$ in $\G$ if and only if $\g{D}=0$ for every disconnected set $D$ in $\G$.

We assume that $D=\tD_{1}\cup\cdots\cup \tD_{r}$ is an arbitrary subset of $V$ that is disconnected in $\G$ and note that, in this case, also every set $E\subseteq \tD_{1}\cup\cdots\cup \tD_{r}$ such that $E\not\subseteq \tD_{i}$ for every $i=1,\ldots, r$ is disconnected in $\G$. Then, if  $X_{\tD_{1}}\ci\cdots\ci X_{\tD_{r}}$ it follows from Corollary~\ref{THM:gamma.eq.zero.and.mutual.indep} that also $\g{D}=0$. On the other hand, if
every element of $\g{}$ corresponding to a disconnected set is equal to zero, then $\g{E}=0$ for every
$E\subseteq \tD_{1}\cup\cdots\cup \tD_{r}$ such that $E\not\subseteq \tD_{i}$ for every $i=1,\ldots, r$ and, by Corollary~\ref{THM:gamma.eq.zero.and.mutual.indep}, this implies that $X_{\tD_{1}}\ci\cdots\ci X_{\tD_{r}}$.

\section{Algorithm for maximum likelihood estimation}\label{SEC:APP-C}
%
Let $n=(n_D)_{D\subseteq V}$ be a vector of cell counts observed under
Multinomial sampling from a binary random vector $X_{V}$
with probability parameter $\p{}>0$.
If we denote by $\psi= N\pi$
the expected value of $n$, where $N=\boldsymbol{1}^\T n$
is the total observed count (sample size)
and $\boldsymbol{1}$ is the unit vector of size $R^{2^{|V|}}$.
We can deal with maximum likelihood estimation of $\p{}$
by considering $n$ as coming from Poisson sampling with parameter
$\psi>0$ and, in this case, we will find
$\boldsymbol{1}^T \hat{\psi}=N$ and $N^{-1}\hat{\psi}=\hat{\pi}$.
Thus, using the reparameterization $\omega = \log\psi$
to remove the positivity constraint on $\psi$,
we can write the log-likelihood function (up to a constant term) as
$$
\ell(\omega;n) = n^\T \omega - \boldsymbol{1}^\T \exp(\omega),
\quad\omega\in\ R^{2^{|V|}}.
$$

The log-mean linear parameter $\gamma$ is obtained from $\omega$ through the
reparameterization $\gamma = \mM^\T \log\{\mZ\exp(\omega)\}$, $\omega\in\ R^{2^{|V|}}$,
so that the linear constraint on $\gamma$ defined by $\mH^\T\gamma=0$
can be transformed into the following non-linear constraint on $\omega$:
$$
g(\omega)=\mH^\T \mM^\T \log\{\mZ\exp(\omega)\} = 0.
$$
Maximum likelihood estimation in the log-mean linear model defined by $\mH$
can thus be formulated as the problem of maximizing the objective function $\ell(\omega;n)$,
with respect to~$\omega$, subject to the constraint $g(\omega)=0$.

A well-known method for the above constrained optimization problem
looks for a saddle point of the Lagrangian function $\ell(\omega;n)+\tau g(\omega)$,
where $\tau$ is a $k$-dimensional vector of unknown Lagrange multipliers,
by solving for $\omega$ and $\tau$ the gradient equation
$$
\frac{\partial\ell(\omega;n)}{\partial\omega}+
\frac{\partial g(\omega)}{\partial\omega}\tau=0
$$
together with the constraint equation $g(\omega)=0$.
If $\hat{\omega}$ is a local maximum of $\ell(\omega;n)$
subject to $g(\omega) = 0$,
and $\partial g(\omega)/\partial\omega$ is a full rank matrix,
then a classical result \citep{bert:1982} guarantees
that there exists a unique $\hat{\tau}$
such that the gradient equation is satisfied by $(\hat\omega,\hat\tau)$.
In the following we assume that the maximum likelihood estimate
of interest is a local (constrained) maximum.

The gradient equation requires that the gradient of $\ell$,
that is, the score vector
$$
s(\omega;n)=\frac{\partial\ell(\omega;n)}{\partial\omega}=n-\exp(\omega),
$$
be orthogonal to the constraining manifold defined by $g(\omega)=0$, that is,
belong to the vector space spanned by the columns of
\begin{eqnarray*}
\mG(\omega)=\frac{\partial g(\omega)}{\partial\omega}&=&
\frac{\partial\{\mZ\exp(\omega)\}}{\partial\omega}
\frac{\partial \log\{\mZ \exp(\omega)\}}{\partial\{\mZ\exp(\omega)\}}\mM \mH\\
&=&\diag\exp(\omega)\,\mZ^\T[\diag\{\mZ\exp(\omega)\}]^{-1}\mM\mH,
\end{eqnarray*}
where $\diag\,v$ is the diagonal matrix with diagonal entries
taken from the vector $v$. We remark that $\mG(\omega)$ has full rank,
for all $\omega\in R^{2^{|V|}}$, because $\mH$ has full rank by construction.

Since no closed-form solution of the system
formed by the gradient and constraint equations is available (in our case)
we resort to an iterative procedure inspired by \citet{ait-sil:1958}
and \citet{lang:1996}. Specifically, we use the Fisher-score-like updating equation
$$
\left[\begin{array}{c}
\omega^{t+1} \\
\tau^{t+1} \\
\end{array}\right]=
\left[\begin{array}{c}
\omega^{t} \\
0 \\
\end{array}\right]+
\left[\begin{array}{cc}
\mF(\omega^{t}) & -\mG(\omega^{t})\\
-\mG(\omega^{t})^\T & 0
\end{array}\right]^{-1}
\left[\begin{array}{c}
s(\omega^t;n) \\
g(\omega^t) \\
\end{array}\right]
$$
to take step $t+1$ of the procedure,
where $\omega^t$ and $\tau^t$ (unused)
are the estimates of $\omega$ and $\tau$ (respectively)
at step $t$, and $\mF(\omega)$ is the Fisher information matrix
$$
\mF(\omega)=-E\left\{\frac{\partial s(\omega;n)}{\partial\omega}\right\}=
-E\{-\diag\exp(\omega)\}=\diag\exp(\omega)
$$
at $\omega\in R^{2^{|V|}}$. The above updating equation is obtained
using a first order expansion of $s(\omega;n)$ and $g(\omega)$ about $\omega^t$;
see \citet{eva-for:2011} for details.

The matrix inversion in the updating equation can be solved block-wise
as follows \citep{ait-sil:1958}:
$$
\left[\begin{array}{cc}
\mF(\omega^{t}) & -\mG(\omega^{t})\\
-\mG(\omega^{t})^\T & 0
\end{array}\right]^{-1}=
\left[\begin{array}{cc}
 \mR & \mQ\\
 \mQ^\T & -\mP^{-1}\\
\end{array}\right],
$$
where
\begin{eqnarray*}
\mP &=& \mG(\omega^t)^\T\mF(\omega^t)^{-1}\mG(\omega^t),\\
\mQ &=& -\mF(\omega^t)^{-1}\mG(\omega^t)\mP^{-1},\\
\mR &=& \mF(\omega^t)^{-1}+\mF(\omega^t)^{-1}\mG(\omega^t)\mQ^\T.
\end{eqnarray*}
Then, introducing the relative score vector
$$
e(\omega^t;n)=\mF(\omega^t)^{-1}s(\omega^t;n)=
\{\diag\exp(\omega^t)\}^{-1}\{n-\exp(\omega^t)\},
$$
the updating equation can be split and simplified as
\begin{eqnarray*}
\tau^{t+1} &=& -\mP^{-1}\{\mG(\omega^t)^{\T}e(\omega^t;n)+g(\omega^t)\},\\
\omega^{t+1} &=& \omega^t+e(\omega^t;n)+\mF(\omega^t)^{-1}\mG(\omega^t)\tau^{t+1},
\end{eqnarray*}
so that the instrumental role of Lagrange multipliers becomes apparent,
and it is clear that the algorithm actually runs in the space of $\omega$.
Notice that the updates take place in the rectangular space
$R^{2^{|V|}}$, so that there is no risk of out of range estimation.

Since the algorithm does not always converge when the starting estimate $\omega^0$
is not close enough to $\hat{\omega}$, it is necessary to
introduce a step size into the updating equation.
The standard approach to choosing a step size in unconstrained optimization problems
is to use a value for which the objective function to be maximized
increases. However, since in our case we are looking for a saddle
point of the Lagrangian function, we need to adjust the standard strategy.
Specifically, \citet{ber:1997} suggests to introduce a step size in the updating equation for $\omega$,
which becomes
\begin{eqnarray*}
\omega^{t+1} &=& \omega^t+
\mathrm{step}^t\{e(\omega^t;n)+F(\omega^t)^{-1}\mG(\omega^t)\tau^{t+1}\},
\end{eqnarray*}
with $0<\mathrm{step}^t\le1$, while the updating equation for $\tau$ is unchanged,
in light of the fact that $\tau^{t+1}$ is computed from scratch at each iteration.
Our choice of $\mathrm{step}^t$ is based on a simple step halving criterion,
which has proven satisfactory for our needs, but more sophisticated criteria are available. At convergence
we obtain $\hat{\gamma}=\mM^\T \log\{\mZ\exp(\hat\omega)\}$
with asymptotic covariance matrix
$$
\mathrm{asy}\,\mathrm{cov}(\hat{\gamma}) = \mJ^\T\mR\mJ,
$$
where $\mJ=\diag\exp(\hat{\omega})\,\mZ^\T[\diag\{\mZ\exp(\hat{\omega})\}]^{-1}\mM$
is the Jacobian of the map $\omega\mapsto\gamma$.

Finally, concerning the choice of the initial estimate $\omega^0$,
we start from the maximum likelihood estimate under the saturated model:
this choice is believed to result in quick convergence,
because it makes the algorithm start close to the data,
and our experience confirms this belief.

\section{Details on the application}\label{SEC:APP-D}
In this section we provide a formal description of some technical details
of our application of log-mean linear models to the data by \citet{coppen:1966}.

Under the connected set Markov property, the bidirected graph in Figure~\ref{FIG.bid.graph}
encodes the marginal independencies  $X_{\{1,2\}} \ci X_{4}$ and
$X_{1} \ci X_{\{3,4\}}$, which are satisfied if and only~if
\begin{eqnarray}\label{EQN:supp001}
\g{\{1,3\}}=\g{\{1,4\}}=\g{\{2,4\}}=\g{\{1,2,4\}}=\g{\{1,3,4\}}=0;
\end{eqnarray}
note that variable coding in uninfluential here.
The directed acyclic graph in Figure~\ref{FIG:DAG.with.latent}
is a possible data generating process for the above bidirected graph model,
because the directed Markov property \citep[\S~3.2.2]{lau:1996} implies, among others, the same marginal independencies and, moreover,
it is associated with the recursive factorization
\begin{eqnarray}\label{EQN:supp002}
\mbox{pr}(X_{V}=x_{V}, U=u)=
\mbox{pr}(x_{2}\mid x_{1}, u)
\mbox{pr}(x_{3}\mid x_{4}, u)
\mbox{pr}(x_{1})
\mbox{pr}(x_{4})
\mbox{pr}(u),
\end{eqnarray}
where $u\in\{\mbox{on},\mbox{off}\}$ and $x_{3}\in\{\mbox{yes},\mbox{no}\}$; see \citet{lau:1996} and \citet{drton2008binary} for details.

We claimed that, if the latent $U$ is a necessary factor for Depression,
that~is, $\mbox{pr}(X_{3}=\mbox{yes}\mid U=\mbox{off})=0$,
then $X_{\{1,2\}}\ci X_{4}\mid \{X_{3}=\mbox{yes}\}$. This follows by noticing that
$$
\mbox{pr}(X_{\{1,2,4\}}=x_{\{1,2,4\}}, U=u\mid X_{3}=\mbox{yes})\propto
\mbox{pr}(X_{\{1,2,4\}}=x_{\{1,2,4\}}, X_{3}=\mbox{yes}, U=u)
$$
so that marginalizing over $U$ one obtains
\begin{eqnarray}\label{EQN:supp005}
\mbox{pr}(X_{\{1,2,4\}}=x_{\{1,2,4\}}\mid X_{3}=\mbox{yes})\propto
\mbox{pr}(X_{\{1,2,4\}}=x_{\{1,2,4\}}, X_{3}=\mbox{yes}, U=\mbox{on}),
\end{eqnarray}
because $\mbox{pr}(X_{\{1,2,4\}}=x_{\{1,2,4\}}, X_{3}=\mbox{yes}, U=\mbox{off})=0$
by the definition of necessary factor. Hence, one can factorize the the right hand side
of (\ref{EQN:supp005}) as in (\ref{EQN:supp002}) and
the required context-specific independence follows immediately from the application of the factorization criterion; see \citet[eqn. (3.6)]{lau:1996}.

We now show that, if the levels of the variable $X3$=Depression are coded
so that $\mbox{yes}=1$, then the bidirected graph model in Figure~\ref{FIG.bid.graph}
satisfies the additional context-specific independence
$X_{\{1,2\}}\ci X_{4}\mid \{X_{3}=1\}$ if and only if,
in addition to (\ref{EQN:supp001}),
it holds that $\gamma_{\{2,3,4\}}=\gamma_{\{1,2,3,4\}}=0$.
To this aim, we first notice that for the conditional distribution of $X_{\{1,2,4\}}\mid \{X_{3}=1\}$ the mean parameter, denoted by $\mu^{(3)}$, has entries
\begin{eqnarray}\label{EQN:supp003}
\mu_{D}^{(3)}=\mbox{pr}(X_{D}
=1_{D}\mid X_{3}=1)=\frac{\mbox{pr}(X_{D}=1_{D},X_{3}=1)}{\mbox{pr}(X_{3}=1)}
=\frac{\mu_{D\cup \{3\}}}{\mu_{\{3\}}}
\end{eqnarray}
for every $D\subseteq \{1,2,4\}$.
From (\ref{EQN:supp003}) it is possible to compute the corresponding
log-mean linear parameter, denoted by $\gamma^{(3)}$,
as a function of $\mu$. In particular, if one computes
$\gamma_{\{1,4\}}^{(3)}$, $\gamma_{\{2,4\}}^{(3)}$, $\gamma_{\{1,2,4\}}^{(3)}$
and then $\gamma_{\{1,3,4\}}$, $\gamma_{\{2,3,4\}}$, $\gamma_{\{1,2,3,4\}}$
by exploiting the factorizations of $\mu$ implied by (\ref{EQN:supp001})
and Theorem~\ref{THM:gamma.eq.zero.and.fact.mu},
then it is straightforward to see that
\begin{eqnarray}\label{EQN:supp004}
\gamma_{\{1,4\}}^{(3)}=\gamma_{\{1,3,4\}}=0,
\qquad
\gamma_{\{2,4\}}^{(3)}=\gamma_{\{2,3,4\}}
\quad\mbox{and}\quad
\gamma_{\{1,2,4\}}^{(3)}=\gamma_{\{1,2,3,4\}}.
\end{eqnarray}
The context-specific independence $X_{\{1,2\}}\ci X_{4}\mid \{X_{3}=1\}$
is a marginal independence in the distribution of $X_{\{1,2,4\}}\mid \{X_{3}=1\}$
and thus, by Theorem~\ref{THM:gamma.eq.zero.and.fact.mu},
it holds if and only if $\gamma_{\{1,4\}}^{(3)}=\gamma_{\{2,4\}}^{(3)}=\gamma_{\{1,2,4\}}^{(3)}=0$.
Therefore, if (\ref{EQN:supp001}) holds true,
it follows from (\ref{EQN:supp004}) that $\gamma_{\{2,3,4\}}=\gamma_{\{1,2,3,4\}}=0$
is a necessary and sufficient condition for  $X_{\{1,2\}}\ci X_{4}\mid \{X_{3}=1\}$
to hold.
\bibliographystyle{chicago}
\bibliography{gamma-ref}
\end{document}